\documentclass[letterpaper, 10 pt, conference]{ieeeconf}  

\IEEEoverridecommandlockouts                              
\usepackage{amssymb}
\usepackage{tikz-cd}
\usetikzlibrary{arrows}
\usepackage{mathrsfs}

\usepackage{amsthm}
\usepackage{amsmath}
\usepackage{graphicx}
\usepackage{subcaption} 
\usepackage{enumerate}

\usepackage{centernot}

\newtheoremstyle{mystyle}
  {}
  {}
  {\itshape}
  {}
  {\bfseries}
  {.}
  { }
  {}

\theoremstyle{mystyle}

\newtheorem{theorem}{Theorem}[section]

\newtheoremstyle{mystyleUpshape}
  {}
  {}
  {}
  {}
  {\bfseries}
  {.}
  { }
  {}

\theoremstyle{mystyleUpshape}
\newtheorem{theoremUpshape}{Theorem}[section]

\newtheorem{remark}[theoremUpshape]{Remark}

\newcommand\xqed[1]{%
  \leavevmode\unskip\penalty9999 \hbox{}\nobreak\hfill
  \quad\hbox{#1}}
\newcommand\exampleEnd{\xqed{$\circ$}}

\usepackage[backend=bibtex,
style=ieee,
bibencoding=ascii,
]{biblatex}
\makeatletter
\newbox\dottedarrow@box
\setbox\dottedarrow@box\hbox
  {%
    \begin{tikzpicture}
      \draw[dotted,->] (0,0) -- (1.5em,0);
    \end{tikzpicture}%
  }
\newcommand*\dottedarrow
  {\relax\ifmmode\expandafter\dottedarrow@m\else\expandafter\dottedarrow@t\fi}
\newcommand*\dottedarrow@t[1][1.5em]
  {\resizebox{#1}{!}{\raisebox{.5ex}{\usebox\dottedarrow@box}}}
\newcommand*\dottedarrow@m[1][]
  {%
    \if\relax\detokenize{#1}\relax
      \mathchoice
        {\dottedarrow@t}
        {\dottedarrow@t}
        {\dottedarrow@t[1.1em]}
        {\dottedarrow@t[0.9em]}%
    \else
      \dottedarrow@t[#1]%
    \fi
  }
\makeatother

\title{\LARGE \bf
First-order plants are globally asymptotically stable\\ under antithetic feedback
}

\author{Wouter Jongeneel
\thanks{\today. The author is with KTH Royal Institute of Technology and Digital Futures, Stockholm, Sweden. Contact: \texttt{\{wouterjo\}}\texttt{@kth.se}. This work is supported by Digital Futures through the project ``\textit{Programmability of Cells}''. This work is part of a project on \textit{The Future of Open Problems in Systems \& Control}. The Lyapunov function \eqref{equ:Lyap:main} we use in this work was found iteratively through Claude Fable 5.1, this note is written by the author.} 
}

\begin{document}

\maketitle
\thispagestyle{empty}
\pagestyle{empty}

\begin{abstract}
We show that first-order linear plants are globally asymptotically stable under antithetic feedback. This is shown by constructing an explicit smooth, convex Lyapunov function and by appealing to a Krasovskii-LaSalle invariance principle. 
\end{abstract}

\section{Introduction}
Antithetic feedback, as pioneered by Briat, Gupta and Khammash~\cite{ref:briat2016antithetic}, is the biological manifestation of integral control. In this note we study the most basic version, that is, with just a single node representing a first-order linear plant. Despite the corresponding ODE being merely a $3$-dimensional bilinear system, global stability remains elusive, \textit{e.g.}, see \cite{ref:olsman2020antithetic} for further comments. Below, we provide a Lyapunov function and conclude on global asymptotic stability, affirmatively.

\section{Antithetic feedback}
We study a dynamical system of the form
\begin{equation}
    \label{equ:ideal:IAF}
    \begin{aligned}
    \dot{x}&= - \gamma x + k z_1\\
    \dot{z}_1&= r - \mu z_1 z_2\\
        \dot{z}_2&= \theta x - \mu z_1 z_2
\end{aligned}
\end{equation}
Here, ``$x$" should be understood as the plant state, whereas ``$z_1$" and ``$z_2$" are the controller species (actuator and sensor, respectively), see Figure~\ref{fig} (I). The parameters $\gamma,k,r,\mu,\theta$ are all elements of $\mathbb{R}_{>0}$. In this sense, the plant is understood as the linear first-order system $\dot{x}=-\gamma x + u$, where the input $u$ is replaced by part of the antithetic feedback scheme. 

To elucidate that~\eqref{equ:ideal:IAF} should be understood as integral control, one may now study $\dot{x}=-\gamma x + k z_1 +d$, for some disturbance $d\in \mathbb{R}_{\geq 0}$. Compactly expressed, we denote the vector field \eqref{equ:ideal:IAF} by $F$. 
Then, consider $\dot{z}_1-\dot{z}_2=r-\theta x$ and see that $F$ has a single equilibrium, being $x^*=r/\theta$, $z_1^*=(\gamma x^*-d)/k$, $z_2^*=rk/(\mu (\gamma x^*-d))$. Note in particular that $x^*$ is independent of the disturbance $d$ and the ``model parameters" $\gamma$ and $k$. Hence, if $\dot{z}_1-\dot{z}_2\to 0$, the interpretation as integral feedback follows. Additionally, see that $\mu z_1^*z_2^*=r$, the feedback term \textit{only} depends on $r$, asymptotically. 

In what follows, we will, for the time being, not work with $d$ (\textit{i.e.}, we set $d=0$) and focus on global stability of the equilibrium. We return to disturbances in Remark~\ref{rem:d}.  

It is well known that the equilibrium is locally exponentially stable under $F$, \textit{e.g.}, linearizing around $(x^*,z_1^*,z_2^*)$ yields
\begin{align*}
    \mathrm{D}F(x^*,z_1^*,z_2^*) = \begin{pmatrix}
        -\gamma & k & 0\\
        0 & -\mu z_2^* & -\mu z_1^*\\
        \theta & -\mu z_2^* & -\mu z_1^*
    \end{pmatrix}
\end{align*}
We can compute the corresponding characteristic polynomial, being $p_1(s) := s^3 + a_2 s^2 + a_1 s + a_0$, for $a_2 = (\gamma + \mu (z_1^*+z_2^*))$, $a_1 = \gamma \mu (z_1^* + z_2^*)$ and $a_0 = k\theta \mu z_1^*$. Then, a straightforward computation and the Routh-Hurwitz stability criterion (\textit{i.e.}, $a_2a_1 - a_0 > 0$) implies that $(x^*,z_1^*,z_2^*)$ is locally exponentially stable under \eqref{equ:ideal:IAF}, for any set of parameters.

To continue, understanding the \textit{global} dynamics of \eqref{equ:ideal:IAF} turns out to be hard. Numerical global studies are provided in \cite{ref:agrawal2019some}, but to the best of our knowledge, no global analytical results, regarding asymptotic stability, exist for this seemingly simple case. This is a short note and we do not intend to be complete, but we do mention the pioneering work \cite{ref:wafi2026boundedness} on proving that all antithetic feedback controllers result in bounded solutions. We also mention that most work is concerned with models more intricate than \eqref{equ:ideal:IAF}, we study \eqref{equ:ideal:IAF} due to links with the p53-Mdm2 interaction \cite{ref:WJCDC26}. Looking at general plants, we point out that especially 2-dimensional plants received more attention, \textit{e.g}, see \cite{ref:olsman2019hard,ref:margaliot2025compact}. For example, it is known that for the 2-dimensional plants, oscillations can occur \cite{ref:briat2016antithetic,ref:margaliot2025compact}. Returning to \eqref{equ:ideal:IAF}, our main result is as follows. 

\begin{figure*}[t!]
    \centering   \includegraphics[scale=0.16]{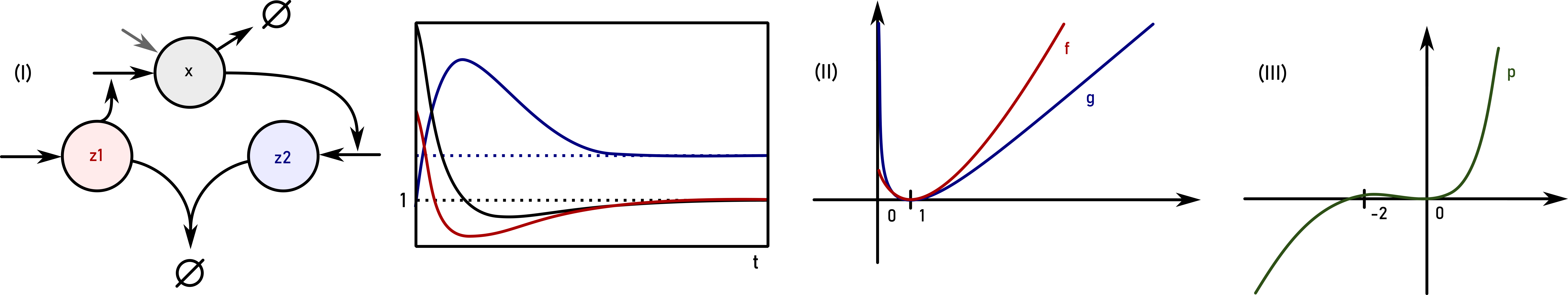}
    \caption{(I): The antithetic feedback motif and a simulation, starting from $(x(0),z_1(0),z_2(0))=(3,2,1)$ with $(k,\gamma,\theta,r,\mu, d)=10^{-1}\cdot(1,2,3,3,2,1)$. (II) and (III): Functions as used in the proof of Theorem~\ref{thm:GAS}.}
    \label{fig}
\end{figure*}

\begin{theorem}[Global asymptotic stability]
\label{thm:GAS}
For any set of parameters $\gamma,k,r,\mu,\theta$, the system \eqref{equ:ideal:IAF} has a unique equilibrium $(x^*,z_1^*,z_2^*)\in \mathbb{R}^3_{>0}$, this equilibrium is globally asymptotically stable. 
\end{theorem}
\begin{proof}
    The fact that \eqref{equ:ideal:IAF} has a unique equilibrium is well known and easy to see. For completeness: $x^*=r/\theta$, $z_1^*=\gamma r/(k\theta)$ and $z_2^*=\theta k/(\mu \gamma )$. It is also straightforward to observe that $\mathbb{R}^3_{\geq 0}$ is invariant under \eqref{equ:ideal:IAF}. In fact, $\mathbb{R}^3_{>0}$ is invariant under \eqref{equ:ideal:IAF}. 

    It pays off to transform \eqref{equ:ideal:IAF} into a system described by just 3 parameters that is still invariant under $\mathbb{R}^3_{>0}$, but now has $(1,1,1)$ as its unique equilibrium. 

    To do so, consider the new coordinates $\widetilde{x}:=x/x^*$, $\widetilde{z}_1:=z_1/z_1^*$ and $\widetilde{z}_2:=z_2/z_2^*$. Also introduce the new parameters $\alpha:=\mu z_1^*$ and $\beta := \mu z_2^*$. Now, exploiting that $r=\mu z_1^*z_2^*$, it follows that 
    \begin{equation}
    \label{equ:ideal:IAF:2}
    \begin{aligned}
    \dot{\widetilde{x}}&= \gamma (\widetilde{z}_1-\widetilde{x})\\
    \dot{\widetilde{z}}_1&= \beta(1 - \widetilde{z}_1 \widetilde{z}_2)\\
        \dot{\widetilde{z}}_2&= \alpha(\widetilde{x} - \widetilde{z}_1 \widetilde{z}_2).
\end{aligned}
\end{equation}
In the remainder of the proof, we first prove Lyapunov stability, then, to be able to appeal to LaSalle's (also due to Barbashin and Krasovskii) invariance principle, we construct compact forward invariant sets that are attracting. 

\textbf{\textit{Lyapunov stability:}}
Towards studying stability, consider the following $C^{\omega}$ function $\widetilde{V}:\mathbb{R}^3_{>0}\to \mathbb{R}_{\geq 0}$ defined through,
\begin{equation}
\label{equ:Lyap}
    \widetilde{V}(\widetilde{x},\widetilde{z}_1,\widetilde{z}_2):=f(\widetilde{x})+(\gamma/2\beta) (\widetilde{z}_1-1)^2+(\gamma/\alpha) g(\widetilde{z}_2),
\end{equation}
where $\widetilde{x} \mapsto f(\widetilde{x} ) := \widetilde{x} \ln(\widetilde{x} )-\widetilde{x}+1$ and $\widetilde{z}_2\mapsto g(\widetilde{z}_2):=\widetilde{z}_2-1-\ln (\widetilde{z}_2)$. Note that both $f$ and $g$ are convex on $\mathbb{R}_{>0}$, with their unique minimizers evaluating to $0$ at $1$, see Figure~\ref{fig} (II). This implies that $\widetilde{V}$ is convex. 

Note that $\lim_{\widetilde{x}\to 0^+} \widetilde{x}\ln (\widetilde{x})-\widetilde{x}+1 = 1$ and not $+\infty$, the same is true for $\lim_{\widetilde{z}_1\to 0^+} (\widetilde{z}_1-1)^2=1$. This means that sublevel sets of $\widetilde{V}$, on $\mathbb{R}^3_{>0}$ are \textit{not} compact. We return to this below.    

Now, the claim is that along solutions of \eqref{equ:ideal:IAF:2}, 
\begin{align}
\label{equ:Vdot}
    \dot{\widetilde{V}}(\widetilde{x},\widetilde{z}_1,\widetilde{z}_2) = - \gamma W(\widetilde{x},\widetilde{z}_1,\widetilde{z}_2)
\end{align}
for the function $W:\mathbb{R}^3_{>0}\to \mathbb{R}$ defined through
\begin{align*}
    W(\widetilde{x},\widetilde{z}_1,\widetilde{z}_2) = (\widetilde{x}-\widetilde{z}_1)\ln (\widetilde{x}) &- (1-1/\widetilde{z}_2) (\widetilde{x}-1) \\
    &+ (\widetilde{z}_1 \widetilde{z}_2 -1)^2 / \widetilde{z}_2. 
\end{align*}
Our aim is to get a grip on the sign of $W$, towards concluding on stability. 

First, see that 
\begin{align*}
    \nabla \widetilde{V}(\widetilde{x},\widetilde{z}_1,\widetilde{z}_2) = \left(\ln (\widetilde{x}),(\gamma/\beta)(\widetilde{z}_1-1),(\gamma/\alpha)(1-1/\widetilde{z}_2) \right). 
\end{align*}
Then, along solutions of \eqref{equ:ideal:IAF:2}, we get that 
\begin{align*}
    \dot{\widetilde{V}}(\widetilde{x},\widetilde{z}_1,\widetilde{z}_2) = \gamma (\widetilde{z}_1-\widetilde{x})\ln (\widetilde{x}) &+ \gamma (\widetilde{z}_1-1)(1-\widetilde{z}_1\widetilde{z}_2) \\ &+ \gamma (1-1/\widetilde{z}_2)(\widetilde{x} - \widetilde{z}_1 \widetilde{z}_2).
\end{align*}
Using $(\widetilde{x} - \widetilde{z}_1 \widetilde{z}_2)=(\widetilde{x}-1)-(\widetilde{z}_1 \widetilde{z}_2-1)$ and $(1-1/\widetilde{z}_2)+(\widetilde{z}_1-1) = (\widetilde{z}_1 \widetilde{z}_2-1)/\widetilde{z}_2$ one may rewrite this into \eqref{equ:Vdot}. 

Next, we claim that indeed $W\geq 0$, with $W(\widetilde{x},\widetilde{z}_1,\widetilde{z}_2)=0$ if and only if $\widetilde{x}=1$ and $\widetilde{z}_1\widetilde{z}_2=1$. 

To that end, let $\widetilde{\zeta}_2:=1/\widetilde{z}_2$ and define $\widetilde{W}(\widetilde{x},\widetilde{z}_1,\widetilde{\zeta}_2):=W(\widetilde{x},\widetilde{z}_1,1/\widetilde{\zeta}_2)$. One obtains: 
\begin{align*}
    \widetilde{W}(\widetilde{x},\widetilde{z}_1,\widetilde{\zeta}_2) = \widetilde{x}\ln(\widetilde{x})-\widetilde{z}_1\ln (\widetilde{x}) &+ (\widetilde{\zeta}_2-1) (\widetilde{x}-1) \\
    &+ (\widetilde{z}_1 -\widetilde{\zeta}_2)^2 / \widetilde{\zeta}_2. 
\end{align*}

Now, for arbitrary but fixed $\widetilde{x}$ and $\widetilde{\zeta}_2$, we see that $\widetilde{z}_1\mapsto \widetilde{W}(\widetilde{x},\widetilde{z}_1,\widetilde{\zeta}_2)$ is quadratic and well-defined on all of $\mathbb{R}$. In fact,
\begin{align*}
   m^*_1(\widetilde{x},\widetilde{\zeta}_2) := \arg \min_{\widetilde{z}_1\in \mathbb{R}}\widetilde{W}(\widetilde{x},\widetilde{z}_1,\widetilde{\zeta}_2)
\end{align*}
is explicitly given by $m^*_1(\widetilde{x},\widetilde{\zeta}_2)=\widetilde{\zeta}_2 ( 1 + (1/2)\ln (\widetilde{x}))$, we shorten it to simply $m_1^*$. That means that $m_1^*\leq 0$ for $\ln(\widetilde{x})\leq -2$. However, then, by convexity, $\widetilde{W}$ must be strictly increasing in $\widetilde{z}_1$, for fixed $\widetilde{x}$ and $\widetilde{\zeta}_2$, on $\mathbb{R}_{> 0}$. As
\begin{align*}
    \widetilde{W}(\widetilde{x},\widetilde{z}_1,\widetilde{\zeta}_2) > \widetilde{W}(\widetilde{x},0,\widetilde{\zeta}_2) = (\widetilde{x}\ln (\widetilde{x}) - \widetilde{x} +1) + \widetilde{\zeta}_2 \widetilde{x} > 0, 
\end{align*}
we can conclude that $\widetilde{W}$ is strictly positive for any triple $(\widetilde{x},\widetilde{z}_1,\widetilde{\zeta}_2)\in \mathbb{R}^3_{>0}$ such that $\ln(\widetilde{x})\leq -2$. 

Then, the other case to check is triples such that $\ln(\widetilde{x})> -2$. We first check $\widetilde{W}(\widetilde{x},m^*_1,\widetilde{\zeta}_2)$. As $m^*_1\ln (\widetilde{x}) = \widetilde{\zeta}_2 \ln(\widetilde{x}) + (\widetilde{\zeta}_2/2)\ln(\widetilde{x})^2$ and $(m^*_1-\widetilde{\zeta}_2)^2/\widetilde{\zeta}_2 = (\widetilde{\zeta}_2/4) \ln(\widetilde{x})^2$ it follows that 
\begin{align*}
    \widetilde{W}(\widetilde{x},\widetilde{z}_1,\widetilde{\zeta}_2) &\geq \widetilde{W}(\widetilde{x},m^*_1,\widetilde{\zeta}_2)\\ &= (\widetilde{x}\ln(\widetilde{x})-\widetilde{x}+1) + \widetilde{\zeta_2}p(\ln(\widetilde{x}))
\end{align*}
for $s\mapsto p(s) := e^s - 1 - s - (1/4)s^2$. One can check that $p(s)\geq 0$ for all $s\geq -2$, with $p$ evaluating to $0$ only at $0$ (on this domain), see Figure~\ref{fig} (III). This closes the second case, showing that $\widetilde{W}\geq 0$ (and thus $W\geq 0$). Looking back at \eqref{equ:Vdot}, we can conclude on Lyapunov stability. 

\textit{\textbf{LaSalle's invariance principle:}}
Now, it readily follows that $W(\widetilde{x},\widetilde{z}_1,\widetilde{z}_2)=0$ if $\widetilde{x}=1$ and $\widetilde{z}_1\widetilde{z}_2=1$. On the other hand, suppose that $W(\widetilde{x},\widetilde{z}_1,\widetilde{z}_2)=0$, that implies $\ln(\widetilde{x})>-2$ and so $\widetilde{W}(\widetilde{x},m^*_1,\widetilde{\zeta}_2)=0$, meaning that $\ln(\widetilde{x})=0$ and thus $\widetilde{x}=1$. Thereby $m^*_1=\widetilde{\zeta}_2$ and thus $m^*_1 \widetilde{z}_2 =1$. Clearly, the largest invariant set satisfying these conditions is $\{(1,1,1)\}$. 

Below we clarify that we can indeed appeal to LaSalle's invariance principle, globally.

Although we cannot exploit compactness of sublevel sets, we can show that 
\begin{align*}
    M_c := \{ (\widetilde{x},\widetilde{z}_1,\widetilde{z}_2)\in \mathbb{R}^3_{>0}: \widetilde{V}(\widetilde{x},\widetilde{z}_1,\widetilde{z}_2)\leq c\}
\end{align*}
 is bounded for any $c>0$. One exploits that $\widetilde{V}$ is separable and $\widetilde{V}\leq c$ implies that $f(\widetilde{x})\leq c$, $(\widetilde{z}_1-1)^2 \leq (2\beta/\gamma) c$ and $g(\widetilde{z}_2)\leq (\alpha/\gamma) c$, \textit{e.g.}, as $g$ is coercive on $\mathbb{R}_{>0}$, $g(\widetilde{z}_2)= (\alpha/\gamma) c$ has two different solutions for $c > 0$. We denote the largest one by $\bar{c}_2$. 

 Then, given such an upper bound for the $\widetilde{z}_2$ components of $M_c$, select scalars $r_1$ and $r_2$ such that $0<r_1<r_2<1/\bar{c}_2$. Next, define $K_c(r_1,r_2) := \{ (\widetilde{x},\widetilde{z}_1,\widetilde{z}_2)\in M_c: \widetilde{x}\geq r_1,\widetilde{z}_1\geq r_2\}$. Then we claim that $K_c(r_1,r_2)$ is a compact set that is both forward invariant and attracting.  

 By construction, $K_c(r_1,r_2)$ is a closed set, relative to $\mathbb{R}^3$, contained in a bounded set of $\mathbb{R}^3_{>0}$, \textit{e.g.}, with $\bar{c}_2$ the bound on the $\widetilde{z}_2$ components. Hence, $K_c(r_1,r_2)$ is compact. 

 The set $M_c$ is forward invariant, so towards forward invariance of $K_c(r_1,r_2)$ we need to show that the lower bounds imposed via $r_1$ and $r_2$ are respected. 

 Let $(t,\xi)\mapsto \widetilde{\varphi}^t(\xi)$ denote a solution to \eqref{equ:ideal:IAF:2}, with $\widetilde{\varphi}^t_i(\xi)$ denoting its $i$-th component.
 We claim that $\widetilde{\varphi}^t_2(\xi)\geq r_2$ for all $t\geq 0$. Suppose not, then there is a $t_2>0$ such that $\widetilde{\varphi}^{t_2}_2(\xi)< r_2$. Define $\tau_2:=\sup\{ t\in [0,t_2]:\widetilde{\varphi}_2^t(\xi)\geq r_2\}$, which is well defined as the flow is continuous and we assume to start in $K_c(r_1,r_2)$, \textit{i.e.}, $\xi\in K_c(r_1,r_2)$. Then, $\widetilde{\varphi}^{\tau_2}_2(\xi)=r_2$ and $\dot{\widetilde{z}}_1(\tau_2)\leq 0$. However, with some abuse of notation, 
 \begin{align}
 \label{equ:bound_z1}
     \dot{\widetilde{z}}_1(\tau_2) = \beta(1-r_2 \widetilde{\varphi}_3^{\tau_2}(\xi))\geq \beta(1-r_2 \bar{c}_2)>0,   
 \end{align}
 which is a contradiction. 
 
 Second, we claim that $\widetilde{\varphi}^t_1(\xi)\geq r_1$ for all $t\geq 0$. Again, suppose not, then there is a $t_1>0$ such that $\widetilde{\varphi}^{t_1}_1(\xi)< r_1$. Define again $\tau_1:=\sup\{ t\in [0,t_1]:\widetilde{\varphi}_1^t(\xi)\geq r_1\}$ such that $\widetilde{\varphi}^{\tau_1}_1(\xi)=r_1$ and $\dot{\widetilde{x}}(\tau_1)\leq 0$, but then
 \begin{align}
 \label{equ:bound_x}
     \dot{\widetilde{x}}(\tau_1) = \gamma(\widetilde{\varphi}_2^{\tau_1}(\xi)-r_1)\geq \gamma(r_2-r_1)>0,
 \end{align}
 which is another contradiction. So far, we can conclude on asymptotic stability on $K_c(r_1,r_2)$ using LaSalle's invariance principle (\textit{e.g.}, see \cite[\S 5.4]{Sastry}) with $K_c(r_1,r_2)$ understood as the domain.

 Next, we claim that $K_c(r_1,r_2)$ is attracting from $M_c$. Exploiting what led to~\eqref{equ:bound_z1} and the comparison principle for ODEs \cite[p. 17]{ref:ArnoldODE}, we obtain:
\begin{align*}
    \widetilde{\varphi}^t_2(\xi)  \geq 1/\bar{c}_2 + (\xi_2-1/\bar{c}_2) e^{-t\beta \bar{c}_2}\quad \forall t\geq 0. 
\end{align*}
This implies that there is a $T_2>0$ such that $\widetilde{\varphi}^t_2(\xi)\geq r_2$ for all $t\geq T_2$. Similarly, exploiting what led to~\eqref{equ:bound_x} and considering $t\geq T_2$, then $\widetilde{\varphi}^t_1(\xi) \geq r_2 + (\widetilde{\varphi}^{T_2}_1(\xi)-r_2)e^{-\gamma (t-T_2)}$. As $r_1<r_2$, this implies that there is a $T_1>T_2$ such that $\widetilde{\varphi}_1^t(\xi)\geq r_1$ for all $t\geq T_1$. 

As $c>0$ was arbitrary, we can conclude on global asymptotic stability. 
\end{proof}

The proof of Theorem~\ref{thm:GAS} constructs a Lyapunov function for \eqref{equ:ideal:IAF} in new coordinates. In the original coordinates, the Lyapunov function $V:\mathbb{R}^3_{>0}\to \mathbb{R}_{\geq 0}$ becomes: 
\begin{equation}
    \label{equ:Lyap:main}
    \begin{aligned}
    V(x,z_1,z_2) := &\left(\frac{x}{x^*}\ln\left(\frac{x}{x^*}\right)-\frac{x}{x^*}+1\right)
    +\frac{\gamma^2}{2\theta k}\left(\frac{z_1}{z_1^*}-1\right)^2\\
&+\frac{\theta k}{\mu r}\left(\frac{z_2}{z_2^*}-1-\ln\left(\frac{z_2}{z_2^*}\right)\right).       
    \end{aligned}
\end{equation}
Perhaps surprisingly, this function is convex, but not a finite-degree polynomial.  

\begin{remark}[Adding disturbances to the analysis]
\label{rem:d}
\upshape{
One may re-introduce $d\in \mathbb{R}_{\geq 0}$, as long as $\theta d< r \gamma $. In this case, one can proceed like in the proof of Theorem~\ref{thm:GAS}, yet with an offset in the $\widetilde{x}$ dynamics, that is, 
\begin{align*}
    \dot{\widetilde{x}} = \gamma_d (\widetilde{z}_1-\widetilde{x}) - \widetilde{d}(\widetilde{x}-1)
\end{align*}
for $\widetilde{d}:=d/x^*$ and $\gamma_d:=\gamma-\widetilde{d}$. Looking at \eqref{equ:Lyap}, with $\gamma$ replaced by $\gamma_d$, see that the addition of $-\widetilde{d}\ln(\widetilde{x})(\widetilde{x}-1)$ to the RHS of \eqref{equ:Vdot} does not cause any complications as the term is strictly negative for any $\widetilde{x}\neq 1$. Regarding the second part of the proof, also \eqref{equ:bound_x} carries over as $r_1<1$ (since $r_1< 1/\bar{c}_2$, $g$ attains its minimum at $1$ and thus $1/\bar{c}_2< 1$). The same is true for the attraction step. Note that for larger disturbances there is no positive equilibrium. 
}
\exampleEnd
\end{remark}

Although Theorem~\ref{thm:GAS} holds for any choice of parameters, one cannot extend the result to any interaction that matches the motif from Figure~\ref{fig} (I). For instance, if we replace the degradation term $-\gamma x$ by the saturated degradation term $-\gamma_1 x/(\gamma_2+x)$, the equilibrium can become unstable and one will observe oscillations. This also means that structural results like \cite{ref:blanchini2017polyhedral} do not apply and a different form of a ``generic" robustness analysis needs to be found. See \cite{ref:proverbio2025robustness} for a recent and more extensive discussion in this direction. 

\section{Discussion and conclusion}
The deceptively simple dynamical system~\eqref{equ:ideal:IAF} seems to resist a simple global stability analysis; as shown in this note, a more involved analysis is possible. Towards the global analysis of more interesting controllers, a more principled analysis would be welcome. Some directions are explored in \cite{ref:margaliot2025compact,ref:wafi2026boundedness}.    

\printbibliography

@article{ref:wafi2026boundedness,
  title={Boundedness of solutions in feedback systems with antithetic controllers},
  author={Wafi, Moh Kamalul and de Oliveira, Arthur CB and Sontag, Eduardo D},
  journal={arXiv preprint arXiv:2604.27290},
  year={2026}
}

@INPROCEEDINGS{ref:WJCDC26,
  author={Jongeneel, Wouter and Scolamiero, Martina and Johansson, Karl H.},
  booktitle={Proc. IEEE CDC}, 
  title={The p53-{M}dm2 interaction through the lens of antithetic feedback}, 
  year={2026},
  volume={},
  number={},
  pages={1-7}}

@inproceedings{ref:agrawal2019some,
  title={Some remarks on robust gene regulation in a biomolecular integral controller},
  author={Agrawal, Deepak K and Marshall, Ryan and Al-Radhawi, M Ali and Noireaux, Vincent and Sontag, Eduardo D},
  booktitle={Proc. IEEE CDC},
  pages={2820--2825},
  year={2019}
}

@article{ref:olsman2019hard,
  title={Hard limits and performance tradeoffs in a class of antithetic integral feedback networks},
  author={Olsman, Noah and Baetica, Ania-Ariadna and Xiao, Fangzhou and Leong, Yoke Peng and Murray, Richard M and Doyle, John C},
  journal={Cell systems},
  volume={9},
  number={1},
  pages={49--63},
  year={2019},
  publisher={Elsevier}
}

@inproceedings{ref:olsman2020antithetic,
  title={Antithetic integral feedback for the robust control of monostable and oscillatory biomolecular circuits},
  author={Olsman, Noah and Forni, Fulvio},
  booktitle={Proc. IFAC World Congress},
  volume={53},
  number={2},
  pages={16826--16833},
  year={2020},
  publisher={Elsevier}
}

@article{ref:proverbio2025robustness,
  title={Robustness and resilience of dynamical networks in biology and epidemiology},
  author={Proverbio, Daniele and Katz, Rami and Giordano, Giulia},
  journal={Foundations and Trends{\textregistered} in Systems and Control},
  volume={12},
  number={2-3},
  pages={112--345},
  year={2025},
  publisher={Emerald Publishing Limited}
}

@article{ref:blanchini2017polyhedral,
  title={Polyhedral {L}yapunov functions structurally ensure global asymptotic stability of dynamical networks iff the {J}acobian is non-singular},
  author={Blanchini, Franco and Giordano, Giulia},
  journal={Automatica},
  volume={86},
  pages={183--191},
  year={2017},
  publisher={Elsevier}
}

@inproceedings{ref:margaliot2025compact,
  title={Compact attractors of an antithetic integral feedback system have a simple structure},
  author={Margaliot, Michael and Wu, Chengshuai and Sontag, Eduardo D},
  booktitle={Proc. IEEE CDC},
  pages={2880--2885},
  year={2025}
}

@book{ref:ArnoldODE,
  title={Ordinary Differential Equations},
  author={Vladimir I. Arnold},
  year={1973},
  publisher={MIT Press},
  address={Cambridge}
}

@article{ref:briat2016antithetic,
  title={Antithetic integral feedback ensures robust perfect adaptation in noisy biomolecular networks},
  author={Briat, Corentin and Gupta, Ankit and Khammash, Mustafa},
  journal={Cell systems},
  volume={2},
  number={1},
  pages={15--26},
  year={2016},
  publisher={Elsevier}
}

@book{Sastry,
  title={Nonlinear Systems},
  author={Shankar Sastry},
  year={1999},
  publisher={Springer},
  address={New York}
}
\end{document}